\documentclass[a4paper,fleqn]{cas-dc}

\usepackage[numbers]{natbib}
\usepackage{amsmath,amssymb,amsthm,mathtools}
\usepackage[capitalize]{cleveref}
\usepackage{tikz}
\usetikzlibrary{arrows,shapes,positioning,calc}
\usepackage{paralist}

\theoremstyle{plain}
\newtheorem{theorem}{Theorem}
\newtheorem{assumption}[theorem]{Assumption}
\newtheorem{corollary}[theorem]{Corollary}
\newtheorem{lemma}[theorem]{Lemma}
\theoremstyle{definition}
\newtheorem{definition}[theorem]{Definition}
\newtheorem{problem}[theorem]{Problem}
\newtheorem{example}[theorem]{Example}

\newcommand{\N}{\mathbb{N}}
\newcommand{\R}{\mathbb{R}}

\newcommand{\policy}{\pi}
\newcommand{\traj}{\xi}

\newcommand{\AP}{\ensuremath{\mathtt{AP}}}
\newcommand{\cT}{\mathcal{T}}
\newcommand{\ldist}{\Upsilon}
\newcommand{\Ldist}{\ensuremath{\mathfrak{D}}}
\newcommand{\trace}{\mathtt{t}}
\newcommand{\labelfunc}{L}

\newcommand{\spec}{\phi}

\newcommand{\until}{\,\mathtt{U}\,}
\newcommand{\Next}{\xspace\bigcirc\xspace}

\newcommand{\RWA}{\ensuremath{\Theta}}
\newcommand{\propC}{\ensuremath{\mathcal{C}}}
\newcommand{\contextRWA}{\ensuremath{K}}
\newcommand{\reachRWA}{\ensuremath{\mathcal{R}}}
\newcommand{\avoidRWA}{\ensuremath{\mathcal{A}}}
\newcommand{\globally}{\xspace\square\xspace}
\newcommand{\finally}{\xspace\Diamond\xspace}
\newcommand{\Mc}{\mathcal{M}}
\newcommand{\Dc}{\mathcal{D}}
\newcommand{\Wc}{\mathcal{W}}
\newcommand{\Tc}{\mathcal{T}}
\tikzset{bplayer0/.style = {draw, thick, shape=ellipse, text width=0.2cm, align=center}}
\tikzset{bplayer1/.style = {draw, thick, shape=rectangle, text width=0.3cm, align=center}}
\newcommand{\bhpos}{1.8}
\newcommand{\bypos}{0.85}
 
\usepackage{hyperref}
\begin{document}
\let\WriteBookmarks\relax
\def\floatpagepagefraction{1}
\def\textpagefraction{.001}
\shorttitle{Context-Triggered Robust MPC for Temporal Logic Specifications}
\shortauthors{Bahari Kordabad et al.}

\title[mode=title]{Context-Triggered Robust MPC for Temporal Logic Specifications}
\author[1]{Arash Bahari Kordabad}[orcid=0000-0001-8931-5372]
\cormark[1]
\cortext[1]{Corresponding author}
\ead{arashbk@mpi-sws.org}

\author[2]{Satya Prakash Nayak}[orcid=0000-0002-4407-8681]
\ead{satya.nayak@ist.ac.at}

\author[1,3]{Sadegh Soudjani}[orcid=0000-0003-1922-6678]
\ead{sadegh@mpi-sws.org}

\author[1]{Anne-Kathrin Schmuck}[orcid=0000-0003-2801-639X]
\ead{akschmuck@mpi-sws.org}

\affiliation[1]{organization={Max Planck Institute for Software Systems (MPI-SWS)}, 
city={Kaiserslautern, 67663},
country={Germany}}

\affiliation[2]{organization={Institute of Science and Technology Austria (ISTA)},
city={Klosterneuburg, 3400},
country={Austria}}

\affiliation[3]{organization={University of Birmingham},
city={Birmingham, B15 2TT},
country={United Kingdom}}
\begin{abstract}
We consider the problem of synthesizing robust feedback controllers for discrete-time linear systems that ensure the satisfaction of context-dependent linear temporal logic specifications in the presence of additive bounded disturbances. Building on existing results that reduce context-triggered temporal logic synthesis to the realization of context-dependent reach-avoid-stay (cRAS) objectives, we focus on the corresponding low-level control synthesis problem. We first employ certificate-based conditions for the almost-sure satisfaction of RAS specifications. Based on these conditions, we propose a switching control architecture that combines robust model predictive control (MPC) with a local invariant controller, and show that the resulting MPC value function serves as a reachability certificate while avoidance is enforced through robust constraints and the stay is enforced via the local controller. To obtain computationally tractable formulations for the resulting robust optimizations, we employ convex duality to reformulate the robust constraints into equivalent deterministic optimization problems, yielding convex quadratic and second-order cone programs for relevant geometric settings. The proposed framework is demonstrated on a robot navigation problem with context-triggered logical switches in both static and moving environments. The results show significantly larger feasible sets than Lyapunov-based approaches, while naturally accommodating dynamic environments and online task reconfiguration.
\end{abstract}

\begin{NoHyper}
\maketitle
\end{NoHyper}

\section{Introduction}\label{sec:intro}

Robust controller synthesis for uncertain systems with continuous state and control spaces under temporal logic specifications is a challenging problem at the intersection of control theory and formal methods~\cite{belta2017formal,lindemann2025formal}. The difficulty stems from the need to simultaneously handle continuous dynamics, state and input constraints, uncertainty, and high-level logical requirements that describe complex tasks beyond classical objectives such as stabilization, tracking, or regulation~\cite{kordabad2025data,kordabad2024control}. The problem becomes even more challenging when the specification depends on the \emph{context}, and the context itself may change over time due to external events or because of the evolution of the controlled system~\cite{nayak2023context,nejati2024context}. In such settings, the controller must not only satisfy a temporal logic objective, but also react online to context-triggered changes in the active task.

To illustrate this setting, we consider a robot navigating in two rooms connected by a sliding door, illustrated in Figure~\ref{fig:env}. The robot is required to move among target regions according to a mode requested by the environment. At the same time, it must avoid walls and, depending on the current configuration, also avoid the door region when the door is closed. The logical context changes over time and affects the active objective of the robot. Moreover, the state of the door, namely whether it is open or closed, changes when the robot enters a particular region. Hence, the motion of the robot and the logical evolution of the task are tightly interconnected. When stochastic disturbances are present in the dynamics, the controller must react to such logical switches while maintaining safe motion and ensuring progress toward the requested target.

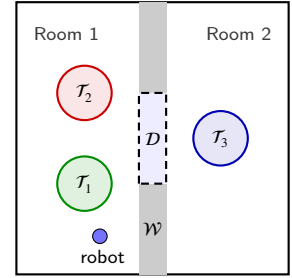
\begin{figure}
\centering
\begin{tikzpicture}[scale=1.2,font=\scriptsize]
	\draw[thick] (0,0) rectangle (3,3);
	\fill[gray!40] (1.35,0) rectangle (1.65,1);
	\fill[gray!40] (1.35,2) rectangle (1.65,3);
	\node[black] at (1.5,0.5) {$\Wc$};
	\draw[densely dashed,thick,fill=blue!7] (1.35,1) rectangle (1.65,2);
	\node at (1.5,1.5) {$\Dc$};
	\filldraw[draw=green!55!black,fill=green!55!black!12,thick] (0.75,1) circle (0.3);
	\node at (0.75,1) {$\Tc_1$};
	\filldraw[draw=red!80!black,fill=red!80!black!10,thick] (0.75,2) circle (0.3);
	\node at (0.75,2) {$\Tc_2$};
	\filldraw[draw=blue!70!black,fill=blue!70!black!10,thick] (2.25,1.5) circle (0.3);
	\node at (2.25,1.5) {$\Tc_3$};
	\filldraw[draw=black,fill=blue!55,rounded corners=0.5pt] (0.92,0.42) circle (0.08);
	\node[anchor=west] at (0.61,0.2) {robot};
	\node[gray!45!black] at (0.55,2.65) {Room~1};
	\node[gray!45!black] at (2.45,2.65) {Room~2};
\end{tikzpicture}
\caption{Schematic of the running example: a robot navigates between two rooms connected by a sliding door $\Dc$ in the wall $\Wc$. The environment requests a mode $\Mc_i$, upon which the robot must reach the corresponding target $\Tc_i$, $i\in\{1,2,3\}$, while avoiding the wall (and the door region when the door is closed). The door opens when the robot enters $\Tc_1$ or $\Tc_3$ and closes when it enters $\Tc_2$.}\label{fig:env}
\end{figure}

A convenient way to address this interaction between logic and dynamics is to separate the synthesis problem into a high-level and a low-level layer. Recent work in~\cite{nayak2023context} has shown that context-triggered temporal logic synthesis can be reduced to strategic reasoning over a game graph, where the high-level controller reacts to context switches using strategy templates~\cite{strategytemplates,PermissiveAssumptions}. 
The strategic choices of the logical controller are then translated into context-triggered reach-avoid-stay (cRAS) objectives for the continuous system: rather than following a fixed plan, the active objective is updated reactively whenever the logical context changes. In the present work, we adopt this viewpoint and focus on the low-level synthesis problem for stochastic systems.

The low-level objectives considered in this paper take the form of RAS tasks: the state should avoid a forbidden set at all times, eventually reach a target set, and remain there thereafter. For stochastic systems with bounded disturbances, such objectives can be characterized through suitable real-valued functions on the state space, which we refer to as \emph{certificates}. More precisely, one may seek an avoidance certificate that separates the unsafe region~\cite{ames2019control} from the state space, a reachability certificate that guarantees the state progress toward the target~\cite{majumdar2024necessary,kordabad2025certificates}, and an invariance certificate that ensures robust stay inside the target set. If these certificates satisfy suitable conditions, then the RAS objective is satisfied with probability one~\cite{majumdar2024necessary,kordabad2026almost}. This observation provides the formal foundation on which our control synthesis method is built.

To synthesize such certificates constructively, we employ robust Model Predictive Control (MPC) alongside a local invariant controller. MPC is an optimization-based control methodology that is particularly well suited to systems with continuous state and input spaces and explicit constraints~\cite{MPCbook}. In the robust setting, it can account for bounded uncertainty by propagating reachable tubes over the prediction horizon and enforcing constraints for all disturbance realizations within those tubes~\cite{langson2004robust}. This makes it a natural candidate for satisfying the avoid and reach parts of an RAS specification. In addition, because the optimization is solved online in a receding-horizon manner, the controller can naturally adapt to changes in the active context and thus react to changes in the environment or in the task specification~\cite{batkovic2020robust}.

\smallskip
\noindent\textbf{Related Work.}
Many existing approaches to temporal-logic-based control synthesis for continuous and hybrid systems either (I) rely heavily on state-space discretization, or (II) neglect context-triggered logical switching. Within (I), abstraction-based methods \cite{tabuada2009_book} construct finite-state models of the continuous dynamics and then solve the resulting discrete synthesis problem using tools from formal methods (e.g.~\cite{Scots,ARCS,Mascot,pFaces,li2018rocs,MajumdarMRSS_CAV23}).
While being able to handle arbitrary reactive (i.e., context-triggered) logical specifications, such approaches typically suffer from conservatism and scalability limitations induced by discretization.
Within (II), optimal control methods are combined with temporal logic constraints \cite{belta2019formal,firouzmand2021robust,nikou2021robust,sadraddini2015robust}, including robust tube-based formulations for constrained dynamical systems and high-level motion planning tasks.
Yet, due to the pre-computation of constraints from logical specifications, context-triggered switching is typically not supported.

Context-triggered abstraction-based control~\cite{nayak2023context}
bridges (I) and (II) by providing a novel abstraction technique that enables strategic reactivity without state-space discretization. This framework has so far been applied to deterministic~\cite{nayak2023context} and stochastic dynamical systems~\cite{nejati2024context,NejatiS_CDC24}, via pre-computed barrier certificates. This paper extends these approaches to robust online optimization, in particular MPC. This enables us to fully exploit the benefits of handling robust uncertainty and complex context-triggered constraints without discretization of the spaces. 
%


\smallskip
\noindent\textbf{Contributions.}
The present paper extends these existing approaches by developing a robust online optimization-based control framework compatible with context-triggered specifications. Our approach does not rely on state-space discretizations, explicitly accounts for bounded disturbances over stochastic linear systems, and provides certificate-based guarantees. 
Our main contributions are as follows:\\
\begin{inparaitem}[$\triangleright$]
    \item We provide certificate-based sufficient conditions for the satisfaction of cRAS objectives under stochastic disturbances.\\
    \item We propose a robust MPC construction for the reach and avoid parts of the specification, and show that the corresponding value function serves as a reachability certificate.\\
    \item We design a switching strategy with a local controller that guarantees the stay part of the specification through robust forward invariance of the target set.\\
    \item We derive tractable deterministic reformulations of the robust optimization problems using convex duality.\\
    \item We demonstrate the proposed approach on a robot navigation problem with context-triggered switching in both static and moving environments, and show that it enlarges the feasible domain compared with the Lyapunov-based approach in~\cite{nayak2023context}.
\end{inparaitem}

The main practical advantage of the resulting framework compared to \cite{nayak2023context, nejati2024context,NejatiS_CDC24} is its receding-horizon structure, which makes the controller naturally responsive not only to context switches but also to changes in the environment, as illustrated in the moving-target scenario considered in this paper. In addition, the tractable robust optimization formulation directly incorporates disturbance effects into the synthesis problem, yielding a less conservative low-level controller. This improvement is reflected in the numerical results, where the proposed method produces a substantially larger feasible domain than the Lyapunov-based approach in~\cite{nayak2023context}, while remaining compatible with the same high-level context-triggered logical framework.

\smallskip
\noindent\textbf{Outline.}
The paper is organized as follows. Section~\ref{sec:pre} provides preliminaries on the considered system and logical specification and delivers the problem definition. Section~\ref{sec:reduction} explains how temporal logic specifications with context switches are reduced to cRAS objectives, following the approach of~\cite{nayak2023context}.  The certificate-based characterization of these objectives and the robust model predictive control construction are developed in Section~\ref{sec:C_RMPC}.  In Section~\ref{sec:tract}, a tractable reformulation of the resulting robust optimization problems is presented via convex duality. Section~\ref{sec:Sim} provides  numerical simulations on static and moving scenarios. Finally, the paper concludes with a summary of the main results in Section~\ref{sec:Conc}.

\section{Preliminaries}\label{sec:pre}
We denote the set of non-negative integers by $\N$ and the set of real numbers by $\R$.  We consider a discrete-time stochastic linear system with bounded state space $X \subset \mathbb R^{n_x}$, input space $U \subseteq \mathbb R^{n_u}$, and bounded disturbance set $W \subset \mathbb R^{n_w}$ equipped with the Borel $\sigma$-algebra $\mathcal B(W)$. Let $(\Omega,\mathscr F,\mathbb P_{\Omega})$ be a probability space, where $\Omega$ is the sample space, $\mathscr F$ is a $\sigma$-algebra on $\Omega$, and $\mathbb P_{\Omega}:\mathscr F\rightarrow[0,1]$ is a probability measure. The disturbance process is modeled as an i.i.d.\ sequence $\{w_k\}_{k\in\mathbb N}$ of random vectors
$
w_k:(\Omega,\mathscr F)\rightarrow (W,\mathcal B(W)).
$
We denote by $\mathbb P_w$ the common law of $w_k$ on $(W,\mathcal B(W))$, namely
$\mathbb P_w(E):=\mathbb P_{\Omega}(w_k^{-1}(E))$ for all $E\in\mathcal B(W)$. Thus, each $w_k$ takes values in $W$ and is sampled according to $\mathbb P_w$. For each $\omega\in\Omega$, the sequence $\{w_k(\omega)\}_{k\in\mathbb N}$ represents a realization of the disturbance process.  The stochastic dynamics are given by
\begin{equation}\label{eq:dyn}
x_{k+1} = A x_k + B u_k + C w_k,
\end{equation}
where $x_k \in X$ is the state, $u_k \in U$ is the control input, and $w_k \in W$ is the disturbance at time step $k$. The matrices $A\in \mathbb{R}^{n_x\times n_x}$, $B\in \mathbb{R}^{n_x\times n_u}$ and $C\in \mathbb{R}^{n_x\times n_w}$ are assumed to be given.

The corresponding nominal (disturbance-free) dynamics are
\begin{equation}\label{eq:dyn_nom}
z_{k+1} = A z_k + B u_k, \quad z_0 = x_0.
\end{equation}

A \emph{control policy} for the stochastic system~\eqref{eq:dyn} is a measurable map $\policy:X\rightarrow U$ that assigns an input to each state. Given a policy $\policy$ and an initial condition $x_0\in X$, a \emph{trajectory} of~\eqref{eq:dyn} under $\policy$ is a sequence $\traj_\pi(\omega) = \{x_0,x_1,\ldots\}$ satisfying $x_{k+1} = Ax_k+B \policy(x_k)+ Cw_k$ for all $k\in\mathbb N$.  

\subsection{Temporal Logic Specifications}\label{subsec:LTL}
We now introduce the logical framework used to specify the desired system behavior. In particular, we employ linear temporal logic (LTL), which we introduce briefly through the running example of Figure~\ref{fig:env}; we refer to~\cite[Chapter~5]{BaierKatoen} for the detailed syntax and semantics.

\smallskip
\noindent\textbf{Atomic propositions.}
The relevant Boolean properties of the system and its environment are captured by \emph{atomic propositions}, grouped into three finite sets. \emph{State propositions} $\AP_S$ describe which region of the state space the robot occupies; in the example, the targets $\Tc_1,\Tc_2,\Tc_3$ and the walls $\Wc$. We identify each state proposition $\cT\in\AP_S$ with the corresponding region $\cT\subseteq X$, so that $\cT$ holds at a state $x$ if and only if $x\in\cT$. \emph{Observation propositions} $\AP_O$ describe information provided by the uncontrolled environment; here, the requested mode $\Mc_1,\Mc_2,\Mc_3$ and the door status $\Dc$ (open or closed). \emph{Control propositions} $\AP_C$ name the finite set of low-level feedback policies that the high-level controller can activate, and are introduced in Section~\ref{sec:reduction}. We write $\AP:=\AP_S\cup\AP_O\cup\AP_C$.

\smallskip
\noindent\textbf{Labels, disturbances, and traces.}
The state propositions induce a \emph{labelling} $\labelfunc\colon X\to 2^{\AP_S}$ with $\cT\in\labelfunc(x)\Leftrightarrow x\in\cT$, which extends to a trajectory $\traj=\{x_0,x_1,\ldots\}$ by $\labelfunc(\traj)=\{\labelfunc(x_0),\labelfunc(x_1),\ldots\}$. The environment is modelled by a \emph{logical disturbance function} $\ldist\colon\N\to 2^{\AP_O}$ that prescribes how the requested mode and door status evolve over time, and we collect all such functions in $\Ldist$. A trajectory $\traj$ together with a disturbance $\ldist$ then produces a \emph{trace} $\trace_{\labelfunc,\ldist}(\traj)=l_0l_1\ldots$ with $l_k=\labelfunc(x_k)\cup\ldist(k)\in 2^{\AP_S\cup\AP_O}$, i.e., the sequence of propositions that hold along the run.

\smallskip
\noindent\textbf{Specifications.}
An LTL specification $\spec$ constrains such traces using the Boolean connectives $\neg,\wedge,\Rightarrow$ and the temporal operators $\Next$ (``next''), $\finally$ (``eventually''), $\globally$ (``always''), and $\until$ (``until''). 
We write $\trace\vDash\spec$ when a trace $\trace$ satisfies $\spec$. For instance, the requirement that the robot always avoids the walls and, whenever a mode $\Mc_i$ is requested persistently, eventually reaches the corresponding target $\Tc_i$, is written as
\begin{equation}\label{eq:examplespec}
	\spec \;=\; \globally\neg\Wc \;\wedge \bigwedge_{i\in\{1,2,3\}} \left(\finally\globally\Mc_i \Rightarrow \finally\globally\Tc_i\right).
\end{equation}

\subsection{Problem Formulation}
Since the controller must react online to the logical context provided by the environment, we consider a \emph{hybrid state-feedback policy}, i.e., a function $\policy_{\mathrm{h}}:\N\times X\times 2^{\AP_O}\to U$ that selects a control input from the current time step $k$, the current state $x_k$, and the currently active observation propositions $\ldist(k)\in 2^{\AP_O}$. 

Under such a hybrid policy $\policy_{\mathrm{h}}$ and a logical disturbance function $\ldist\in\Ldist$, a trajectory of~\eqref{eq:dyn} from $x_0\in X$ is a sequence $\traj_{\pi_{\mathrm{h}}}(\omega)=\{x_0,x_1,\ldots\}$ with $x_{k+1}=Ax_k+B\policy_{\mathrm{h}}(k,x_k,\ldist(k))+Cw_k$, which generates the trace $\trace_{\labelfunc,\ldist}(\traj_{\pi_{\mathrm{h}}}(\omega))$ over $\AP_S\cup\AP_O$.

\begin{problem}\label{prob:MainProb}
Given the stochastic system~\eqref{eq:dyn} with an initial state $x_0\in X$, a labelling function $\labelfunc$, and an LTL specification $\spec$ over $\AP_S\cup\AP_O$, synthesize a hybrid state-feedback policy $\policy_{\mathrm{h}}$ such that, for all logical disturbance functions $\ldist\in\Ldist$,
\begin{equation*}
     \mathbb P_{\Omega}\bigl(\{\omega\in\Omega \mid \trace_{\labelfunc,\ldist}(\traj_{\pi_{\mathrm{h}}}(\omega))\vDash \spec\}\bigr)=1.
\end{equation*}
\end{problem}
That is, Problem~\ref{prob:MainProb} asks for a hybrid policy that drives the stochastic system~\eqref{eq:dyn} to almost surely satisfy the specification $\spec$.

\section{LTL Specifications to RAS Specifications}\label{sec:reduction}
In this section, we summarize how Problem~\ref{prob:MainProb} can be reduced to synthesizing controllers for a set of \emph{context-triggered reach-avoid-stay (cRAS) tasks}. We follow the context-triggered hybrid control framework of~\cite{nayak2023context} and recall only the parts needed for the remainder of the paper, referring the reader to~\cite{nayak2023context} for the formal construction and its correctness proofs. The single object that the low-level synthesis in the subsequent sections operates on is the RAS objective formalized in this section.

\subsection{High-level Logical Synthesis}
The reduction rests on the classical correspondence between temporal logic and two-player games (see, e.g.,~\cite{LTLgamesBook}): every LTL formula $\spec$ over $\AP_S\cup\AP_O$ can be translated into an equivalent two-player game between the \emph{controller} and the \emph{environment} played on a finite graph. Each vertex carries a label, namely a set of propositions describing the current \emph{logical context} (e.g., the requested mode and the door status), and the two players move along the edges of the graph, with the controller resolving the state propositions $\AP_S$ and the environment resolving the observation propositions $\AP_O$. Solving this game with classical methods from reactive synthesis \cite{finkbeiner2016synthesis} results in a (winning) controller strategy which ensures that $\spec$ is satisfied no matter how the environment chooses its moves, i.e., changes the context.

However, rather than committing to such a single winning strategy, the framework of~\cite{nayak2023context} computes a \emph{strategy template}~\cite{strategytemplates,AnandNS24} instead: a compact collection of local edge constraints that characterizes a whole family of winning strategies. Intuitively, at each vertex, such a template marks some outgoing edges as eventually \emph{required} and others as \emph{prohibited}, while leaving the controller free to choose among the remaining ones. This freedom is precisely what allows the high-level logical decisions to be elegantly delegated to the continuous control layer, as described next.

\subsection{Reach-Avoid-Stay Specifications}

Strategy templates optimistically assume that state propositions in $\AP_S$ can be toggled instantaneously. In reality, however, a state proposition becomes true only once the continuous system~\eqref{eq:dyn} actually drives its state into the corresponding region while the active context persists. Consequently, the edge constraints that a strategy template imposes at a given vertex translate, per context, into a \emph{RAS task} for the continuous system: the vertex label fixes the \emph{context}, the regions associated with the required successors form the set to be \emph{reached}, and the regions associated with the prohibited successors form the set to be \emph{avoided}. 
This translation is illustrated in the following example and formalized afterwards.

\begin{figure}
\centering
\begin{tikzpicture}
	\node[bplayer1,label={[align=center]below:$\{\Tc_2\}$}] (b) at (-1.5*\bhpos,0) {$b$};
	\node[bplayer0,label={[align=center]left:$\{\Mc_1\}$}] (c) at (0, \bypos) {$c$};
	\node[bplayer0,label={[align=center]below:$\{\Mc_1,\Dc\}$}] (d) at (0, -\bypos) {$d$};
	\node[bplayer1,label={[align=center]below:$\{\Tc_1\}$}] (e) at (1.5*\bhpos, 0) {$e$};
	\node[bplayer1,label={[align=center]right:$\{\Wc\}$}] (f) at (0, 0) {$a$};

	\path[->,thick] (b) edge[bend left =-10] (d);
	\path[->,thick] (c) edge[dashed,blue!80] (b) edge[bend left =10] (e) edge[dotted,red, line width = 0.05cm] (f);
	\path[->,thick] (d) edge (e) edge[dotted,red, line width = 0.05cm] (f) edge[dashed,blue!80,bend left =-10] (b);
	\path[->,thick] (e) edge[bend left =10] (c);
\end{tikzpicture}
\caption{A part of the game for the robot navigation example. The controller chooses at circled vertices and the environment chooses at squared vertices; each vertex label gives the logical context active at that vertex. The strategy template prohibits the red dotted edges and eventually prohibits the blue dashed edges, while the controller may freely use the remaining edges.}\label{fig:gamegraph}
\end{figure}
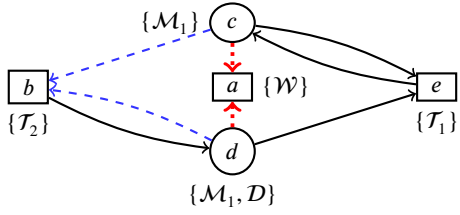

\begin{example}\label{ex:gamegraph}
Consider the robot navigating two rooms connected by a sliding door from~\cite{nayak2023context}, which is the running example of this paper (as explained in Section~\ref{sec:intro}). The robot must reach the target $\Tc_i$ of the mode $\Mc_i$ currently requested by the environment while avoiding the walls $\Wc$. 
A part of the corresponding game is shown in Figure~\ref{fig:gamegraph}. At the controller vertex $d$, whose label $\{\Mc_1,\Dc\}$ encodes the context ``mode $\Mc_1$ requested and door closed'', the strategy template prohibits the edge to $a$ (where the wall $\Wc$ would be hit) and eventually prohibits the edge to $b$ (which would let the robot cycle without ever reaching the requested target), so that the controller eventually commits to the edge to $e$ (labeled $\{\Tc_1\}$). This induces the context-dependent RAS task: \emph{whenever the context $\{\Mc_1,\Dc\}$ persists, reach the target $\Tc_1$ while avoiding the wall $\Wc$ and the target $\Tc_2$}. Such a task is formalized next.\qed
\end{example}

\begin{definition}\label{def:rwa}
    A \emph{context-dependent reach-avoid-stay objective} (cRAS) is defined as a triple $\RWA := (\contextRWA,\reachRWA,\avoidRWA)$  where $\contextRWA\subseteq \AP_O$ is the \emph{context}, $\reachRWA\in 2^{\AP_S}$ is the target set and $\avoidRWA\in 2^{\AP_S}$ is the avoid set. 
    A control proposition $\propC\in\AP_C$ is said to \emph{implement} the cRAS $\RWA$ if the following specification given by $\spec_\propC$ holds with probability one under the closed-loop system induced by $\propC$:
	\begin{equation}
		\spec_\propC:=\globally\Big (\globally (\propC ~\wedge~\contextRWA) \Rightarrow \finally\globally \reachRWA ~ \wedge ~  \globally \neg\avoidRWA\Big).
        \label{eq:ras}
	\end{equation}
\end{definition}

Intuitively, the cRAS objective $\RWA$ requires that, whenever the context $\contextRWA$ is active persistently, the system eventually reaches the target set $\reachRWA$ and stays there while avoiding the obstacle set $\avoidRWA$ at all times. Here, the control proposition 
$\propC\in\AP_C$ identifies one feedback policy that the high-level controller may trigger to actuate the stochastic system~\eqref{eq:dyn}.

\subsection{Reduction to Context-Triggered RASs}\label{subsec:hybrid}
Since the environment may change the observation propositions at any time, and since a context also changes once the reach-part of the active cRAS has been fulfilled, the active cRAS is switched online, and the overall controller is inherently \emph{hybrid}. 
This renders the reduction sound: realizing each generated cRAS is sufficient to solve the original synthesis problem, as summarized next.

\begin{lemma}[\cite{nayak2023context}]\label{lemma:reduction}
For Problem~\ref{prob:MainProb}, the framework of~\cite{nayak2023context} produces a finite set of context-dependent RAS objectives $\{\RWA_1,\ldots,\RWA_m\}$, each triggered by its context. If each $\RWA_j$ is implemented by some control proposition $\propC_j\in\AP_C$ in the sense of Definition~\ref{def:rwa}, then the resulting hybrid state-feedback policy $\policy_{\mathrm{h}}$ solves Problem~\ref{prob:MainProb}.
\end{lemma}

In view of Lemma~\ref{lemma:reduction}, the high-level synthesis returns a finite set of context-dependent RAS objectives, and it remains to realize each of them on the stochastic system~\eqref{eq:dyn}. The rest of the paper therefore focuses on the central problem of synthesizing, for a single context-dependent RAS objective $\RWA=(\contextRWA,\reachRWA,\avoidRWA)$, a feedback policy that implements it almost surely.

\section{Realizing RASs via Certificates}\label{sec:C_RMPC}

As stated in Section~\ref{sec:reduction}, the overall synthesis problem in Problem~\ref{prob:MainProb} can be reduced to synthesizing control policies for a finite set of cRAS specifications.
Hence, we focus on the problem of realizing a single RAS specification of the form $\finally\globally \reachRWA ~ \wedge ~  \globally \neg\avoidRWA$ (as in \eqref{eq:ras}) for a stochastic linear system.

\subsection{Certificate-Based Characterization of RASs}
To characterize when the RAS specification holds, we adopt certificate-based conditions from the literature~\cite{majumdar2024necessary,ames2019control}. These conditions provide sufficient guarantees for almost-sure satisfaction of~\eqref{eq:ras}. Before stating the theorem, we provide the following assumption that characterizes robust forward invariance of the target set and is therefore essential for the \textit{stay} part of the specification. 

\begin{assumption}[Robustly Invariant Target Sets]\label{assum:target_invariance}
The target set $\mathcal R$ is robustly forward invariant, i.e., there exists a feedback controller $\kappa_{\mathcal R}: X\rightarrow U$ such that for all $x\in\mathcal R$ and all $w\in W$,
\[
Ax+B\kappa_{\mathcal R}(x)+Cw \in \mathcal R.
\]
\end{assumption}
Indeed, if no feedback controller can keep the state inside the target set $\mathcal R$ under all admissible disturbances, then the stay requirement cannot be guaranteed after the target has been reached. In the proposed framework, the corresponding invariant controller is constructed through the local optimization problem introduced later in the paper, together with its tractable reformulation. Therefore, this assumption should be viewed as a fundamental requirement for satisfying the specification rather than as an additional restriction of the proposed approach.

\begin{theorem}[RAS Certificates]
\label{thm:RWA}
For the stochastic system~\eqref{eq:dyn}, suppose there exist functions $h, g : \mathbb{R}^{n_x} \to \mathbb R,$ a locally bounded function $r : \mathbb{R}^{n_x} \to \mathbb R$, and constants $\varepsilon > 0$ and $\delta > 0$, such that for every $x \in \mathbb{R}^{n_x}$ there exists a control feedback law $u=\kappa(x) \in  U$ satisfying the following conditions:

\medskip
\noindent\textnormal{\textbf{(Avoidance)}}
\begin{subequations}
\begin{align}
& g(x) < 0, && \forall x \in \mathcal A, \label{eq:AV1}\\
& g(x) \geq 0, && \forall x \in \mathbb{R}^{n_x} \setminus \mathcal A, \label{eq:AV2}\\
& g( A x + B u + C w) \ge 0,
&& \forall x \in \mathbb{R}^{n_x} \setminus \mathcal A,\ \forall w \in  W, \label{eq:AV3}
\end{align}
\end{subequations}

\medskip
\noindent\textnormal{\textbf{(Reach)}}
\begin{subequations}
\begin{align}
& \mathbb P_w\left(
r( A x + B u + C w) - r(x) \le -\delta
\right) \ge \varepsilon,\nonumber\\ &\qquad\qquad\qquad  \forall x \!\in\! \mathbb{R}^{n_x} \!\!\setminus\!\! (\mathcal A \!\cup\! \mathcal R), \label{eq:RE1}\\
& r(x) \le 0 \Rightarrow x \in \mathcal R, \label{eq:RE2}
\end{align}
\end{subequations}

\medskip
\noindent\textnormal{\textbf{(Stay)}}
\begin{subequations}
\begin{align}
& h(x) \le 0, && \forall x \in \mathcal R, \label{eq:ST1}\\
& h(x) > 0, && \forall x \in \mathbb{R}^{n_x} \setminus \mathcal R, \label{eq:ST2}\\
& h( A x + B u + C w) \le 0,
&& \forall x \in \mathcal R,\ \forall w \in W.\label{eq:ST3}
\end{align}
\end{subequations}

Then, for every initial state $x_0 \in \mathbb{R}^{n_x} \setminus \mathcal A$ and with the state-feedback policy $\kappa(x)$, the RAS specification~\eqref{eq:ras} holds almost surely.
\end{theorem}

\begin{proof}
Conditions~\eqref{eq:AV1}--\eqref{eq:AV3} ensure avoidance of the set $\mathcal A$.
In particular,~\eqref{eq:AV1}--\eqref{eq:AV2} characterize the sign of the avoidance function $g$, while~\eqref{eq:AV3} guarantees non-negativity of $g$ is preserved along trajectories for all disturbances. Hence, starting from any state $x \in \mathbb{R}^{n_x} \setminus \mathcal A$, the system never enters the avoid set $\mathcal A$.

For states outside $\mathcal A \cup \mathcal R$, the reach conditions ~\eqref{eq:RE1}--\eqref{eq:RE2} ensure a uniform
positive probability of a decrease in the function $r$. Specifically,~\eqref{eq:RE1} enforces a uniform positive probability of decrease in the reach function $r$ outside $\mathcal A \cup \mathcal R$, while~\eqref{eq:RE2} associates the nonpositive level set of $r$ with the target set.
Since $r$ is locally bounded and the disturbance sequence is i.i.d., the probabilistic decrease condition~\eqref{eq:RE1} implies that the event of never reaching the set $\mathcal R$ has probability zero in the limit.
A detailed proof of this claim is provided in~\cite{majumdar2024necessary}.
Hence, the system reaches $\mathcal R$ almost surely.

Finally, once the system enters $\mathcal R$, the stay conditions~\eqref{eq:ST1}--\eqref{eq:ST3} imply that
$h(x)\le 0$ is preserved under all disturbances.
Therefore, the system remains in $\mathcal R$ for all future times with probability one.
\end{proof}
Note that~\eqref{eq:RE1} is a probabilistic Lyapunov-type decrease condition. Furthermore,~\eqref{eq:RE2} implies that the target set $\mathcal R$ contains the non-positive sublevel set of $r$. The argument in~\cite{majumdar2024necessary} proceeds by observing that, since $r$ is locally bounded, every bounded sublevel set can be partitioned into finitely many layers according to the value of $r$.
Condition~\eqref{eq:RE1} guarantees that, whenever the state is outside $\mathcal A\cup\mathcal R$, there is a uniform positive probability $\varepsilon$ of decreasing the certificate by at least $\delta$. As a consequence, the probability of remaining forever in any non-terminal layer is zero. Since only finitely many layers exist within a bounded sublevel set, the trajectory almost surely exits each layer and eventually reaches the non-positive sublevel set of $r$. By~\eqref{eq:RE2}, this implies eventual entry into the target set $\mathcal R$ with probability one.

\begin{assumption}\label{assum:cvx}
   We assume that the avoid and the target sets admit convex representations.
Specifically, we consider
\begin{equation*}
    \mathcal A = \{x \in \mathbb R^{n_x} : g(x) < 0\},
\,
\mathcal R = \{x \in \mathbb R^{n_x} : h(x) \le 0\},
\end{equation*}
where $g$ and $h$ are proper, closed, convex functions. Furthermore, the disturbance set $ W$ is compact and convex, and satisfies $\mathbb P_w(\|w\| \le \eta) > 0$ for all sufficiently small $\eta>0$.

\end{assumption}
Under this construction, conditions~\eqref{eq:AV1},~\eqref{eq:AV2},~\eqref{eq:ST1}, and~\eqref{eq:ST2} are satisfied by definition, while the remaining conditions will be enforced through the control synthesis developed in the next section.

\subsection{Enforcing RASs via Robust MPC}\label{subsec:C_RMPC}

In this section, we design a control strategy that enforces the RAS specification using a combination of robust model predictive control (MPC) and a local switching law. Robust MPC is used to explicitly enforce the avoidance constraints while driving the system toward the target set, and we complement it with a local controller that ensures robust forward invariance of the target set. The resulting switching strategy satisfies the Reach, Avoid, and Stay conditions of Theorem~\ref{thm:RWA}.

We fix a cRAS objective \(\RWA=(\contextRWA,\mathcal R,\mathcal A)\), where \(\contextRWA\subseteq \AP_O\) denotes the logical context under which the objective is active, \(\mathcal R\) is the target set and \(\mathcal A\) is the avoid set. For readability, we keep the dependence on \(\RWA\) explicit in the notation below.

We consider the following robust MPC problem:
\begin{subequations}\label{eq:RMPC:orig}
\begin{align}
V_\RWA(x) = \min_{z,u}\quad
& T_\RWA(z_N) + \sum_{i=0}^{N-1} \ell_\RWA(z_i,u_i) \\
\text{s.t.}\quad
& \forall i\in\{0,\ldots,N-1\}:\nonumber \\
& z_{i+1} = A z_i + B u_i, \label{eq:cons1:org} \\
& u_i \in  U, \quad z_{i}\in \mathcal X^\RWA_{i}, \label{eq:cons2:org} \\
& z_N \in \mathcal X^\RWA_{\mathrm f}, \quad z_0 = x. \label{eq:cons3:org}
\end{align}
\end{subequations}
Here, \(T_\RWA\) is the terminal cost, \(\ell_\RWA\) is the stage cost, \(N\) is the prediction horizon, \(V_\RWA\) is the MPC value function and \(\mathcal X^\RWA_{\mathrm f}\) is the terminal set associated with the fixed RAS objective \(\RWA\). The constraints in~\eqref{eq:cons1:org}--\eqref{eq:cons3:org} enforce the nominal system dynamics over the prediction horizon, the admissible input constraint, robust avoidance, and the terminal condition, respectively.

The sets \(\mathcal X^\RWA_{i}\) in constraint~\eqref{eq:cons2:org} encode robust avoidance by requiring that the avoidance function \(g\) remains nonnegative for all disturbance realizations over the horizon. They are defined as
\begin{equation*}
\mathcal X^\RWA_{i}
:=
\Bigl\{
x\in X
\,\Big|\,
\min_{\bar w\in \bar{ W}}
g(x + C_{i-1} \bar w)\ge 0
\Bigr\}, 
\end{equation*}
for all $i\in \{0,1,\ldots, N\}$, where \(\bar{ W}:= W\times \cdots \times  W\) (\(N\) times) and
\[
C_i
:=
\bigl[
A^{i}C \;\; A^{i-1}C \;\; \cdots \;\; C \;\; 0 \;\; \cdots \;\; 0
\bigr]
\in \mathbb R^{n_x\times (N n_w)},
\]
for all $i\in \{0,1,\ldots, N-1\}$ and $C_{-1}:=0$. The optimization~\eqref{eq:RMPC:orig} is solved in a receding-horizon manner by applying the first optimal input. The corresponding MPC feedback law is
\[
\kappa^\RWA_{\mathrm{MPC}}(x) = u_0^\star,
\]
where \(u_0^\star\) is the first optimal input of~\eqref{eq:RMPC:orig}.

To promote reachability of the target set \(\mathcal R\), we select quadratic stage and terminal costs centered at a reference point \(x_{\mathrm{ref}}\in\mathcal R\) as an interior point of $\mathcal R$ (for instance, the center of the target set). More specifically, we define
\begin{equation}\label{eq:stage_terminal}
\ell_\RWA(x,u)
=
\|x - x_{\mathrm{ref}}\|^2_{Q_\ell}
+
\|u-u_{\mathrm{ref}}\|^2_{R_\ell},
\,
T_\RWA(x)
=
\|x - x_{\mathrm{ref}}\|^2_{Q_T},
\end{equation}
where \(u_{\mathrm{ref}}=\kappa_{\mathcal R}(x_{\mathrm{ref}})\), \(Q_\ell\succeq 0\), and \(Q_T,R_\ell\succ 0\).

We next establish that the robust MPC problem~\eqref{eq:RMPC:orig} is recursively feasible under nominal closed-loop evolution, that is, when the state evolves according to \(x_{k+1}=Ax_k+B\kappa^\RWA_{\mathrm{MPC}}(x_k)\) under the fixed objective \(\RWA\). This property, referred to as nominal recursive feasibility, ensures that feasibility at the current time step implies feasibility at all subsequent time steps under the nominal dynamics and provides the reachability certificate introduced in Theorem~\ref{thm:RWA}. To guarantee this property, we impose the following standard terminal condition.

\begin{assumption}\label{Assum:terminal}
The terminal set \(\mathcal X^\RWA_{\mathrm f}\) is nonempty, it satisfies \(\mathcal X^\RWA_{\mathrm f}\subseteq \mathcal X^\RWA_N\), and there exists a terminal feedback control \(\kappa_f^\RWA:X\to U\) such that, for all \(x\in \mathcal X^\RWA_{\mathrm f}\),
\begin{subequations}
\begin{align}
T_\RWA(Ax+B\kappa_f^\RWA(x)) - T_\RWA(x)
&\le -\ell_\RWA(x,\kappa_f^\RWA(x)), \label{eq:decrease} \\
Ax+B\kappa_f^\RWA(x)
&\in \mathcal X^\RWA_{\mathrm f}. \label{eq:inv}
\end{align}
\end{subequations}
\end{assumption}

Assumption~\ref{Assum:terminal} is standard in the MPC literature (see, e.g.,~\cite{MPCbook}). Inequality~\eqref{eq:decrease} imposes a terminal decrease condition and plays the role of a Lyapunov-type inequality. In the quadratic setting considered in this paper, it typically reduces to a Riccati-type condition. The invariance condition~\eqref{eq:inv} requires the terminal set to be forward invariant and is the key ingredient ensuring nominal recursive feasibility of the MPC problem. 

In practice, the terminal set \(\mathcal X^\RWA_{\mathrm f}\) is typically selected as a sufficiently small neighborhood of the target set \(\mathcal R\) (or of the reference point \(x_{\mathrm{ref}}\in\mathcal R\)). The existence of a terminal feedback controller satisfying~\eqref{eq:decrease}--\eqref{eq:inv} depends on the control authority available to the system. In particular, for fully actuated systems, such terminal controllers can often be constructed for comparatively large terminal sets. More generally, Assumption~\ref{Assum:terminal} can be viewed as requiring the terminal region to admit a locally stabilizing and invariant control law.

Note that we only require nominal recursive feasibility of the MPC problem here. Indeed, the avoidance condition is enforced robustly through the sets $\mathcal X_i^\RWA$, while the stay condition is handled separately through the local invariant controller. Hence, recursive feasibility is used to construct a shifted feasible solution and establish a decrease property of the MPC value function along the nominal closed-loop dynamics. This decrease property is the key ingredient in constructing the reachability certificate required by Theorem~\ref{thm:RWA}. The probabilistic decrease condition of the actual stochastic system is then obtained by combining the nominal decrease with the disturbance assumption in~\ref{assum:cvx}.

\begin{theorem}\label{thm:MPC}
Under Assumptions~\ref{assum:cvx} and~\ref{Assum:terminal}, the following hold for the MPC~\eqref{eq:RMPC:orig}:
\begin{enumerate}
    \item for every initial condition \(x_0 \in \mathcal X^\RWA_{\mathrm f}\), the MPC problem~\eqref{eq:RMPC:orig} is feasible at time \(0\);
    \item the MPC problem~\eqref{eq:RMPC:orig} is nominal recursively feasible;
    \item the avoidance function \(g\) satisfies the Avoid conditions of Theorem~\ref{thm:RWA};
    \item the value function \(V_\RWA\) admits a decrease property outside \(\mathcal R\) that allows the construction of a reachability certificate \(r\) satisfying conditions~\eqref{eq:RE1}--\eqref{eq:RE2}.
\end{enumerate}
\end{theorem}

\begin{proof}
We first prove that
\[
\mathcal X^\RWA_{i+1}\subseteq \mathcal X^\RWA_i,
\qquad
i\in\{0, 1,\ldots,N-1\}.
\]
Let \(x\in \mathcal X^\RWA_{i+1}\) be fixed, and suppose that an optimal solution of
\[
\min_{\bar w\in \bar{ W}} g(x + C_{i-1}\bar w)
\]
is given by \(\tilde w = (\tilde w_0,\dots,\tilde w_{N-1})\). Define \(\hat w=(0,\dots,\tilde w_{N-1})\). Then
\begin{align*}
\min_{\bar w\in \bar{ W}} g(x+C_{i-1}\bar w)
&= g(x+C_{i-1}\tilde w) \\
&= g(x+C_i\hat w) \\
&\ge \min_{\bar w\in \bar{ W}} g(x+C_i\bar w).
\end{align*}
Since \(x\in \mathcal X^\RWA_{i+1}\), the right-hand side is nonnegative, implying \(x\in \mathcal X^\RWA_i\).

\smallskip
1) For the initial feasibility for every \(x_0\in \mathcal X^\RWA_{\mathrm f}\), fix any \(x_0\in \mathcal X^\RWA_{\mathrm f}\). By Assumption~\ref{Assum:terminal}, for each \(x\in \mathcal X^\RWA_{\mathrm f}\) there exists a control input \(\kappa_f^\RWA(x)\in U\) such that \(Ax+B\kappa_f^\RWA(x)\in \mathcal X^\RWA_{\mathrm f}\). Hence, defining
\[
z_0=x_0,
\qquad
z_{i+1}=Az_i+B\kappa_f^\RWA(z_i),
\qquad
i=0,\dots,N-1,
\]
yields a nominal state sequence satisfying \(z_i\in\mathcal X^\RWA_{\mathrm f}\) for all \(i\in 0,\dots,N\). Since \(\mathcal X^\RWA_{\mathrm f}\subseteq \mathcal X^\RWA_i\) for all \(i\in 0,\dots,N-1\), this sequence satisfies the state constraints of the MPC problem, and the associated inputs \(\kappa_f^\RWA(z_i)\in U\) satisfy the input constraints. Therefore, the MPC problem is feasible at time \(0\) for every \(x_0\in \mathcal X^\RWA_{\mathrm f}\).

\smallskip
2) To show recursive feasibility, assume that \(\{x=z_0^\star,z_1^\star,\dots,z_N^\star\}\) and \(\{u_0^\star,u_1^\star,\dots,u_{N-1}^\star\}\) is an optimal solution with \(z_0^\star=x\). The successor nominal state is then \(z_1^\star=A x + B \kappa^\RWA_{\mathrm{MPC}}(x)\). Consider the shifted sequence \(\tilde z_i=z_{i+1}^\star\) for \(i\in 0,\dots,N-1\), \(\tilde u_i=u_{i+1}^\star\) for \(i=0,\dots,N-2\), \(\tilde u_{N-1}=\kappa_f^\RWA(z_N^\star)\), and \(\tilde z_N = A z_N^\star + B\kappa_f^\RWA(z_N^\star)\). By feasibility of the original solution and \(\mathcal X^\RWA_{i+1}\subseteq \mathcal X^\RWA_i\), we have \(\tilde z_i\in\mathcal X^\RWA_i\) for \(i=0,\dots,N-1\), and from Assumption~\ref{Assum:terminal}, \(\tilde z_N\in\mathcal X^\RWA_N\). Hence, the shifted sequence is feasible at \(z_1^\star\).

\smallskip
3) We next show that the avoidance function \(g\) satisfies the Avoid conditions \eqref{eq:AV1}--\eqref{eq:AV3} of Theorem~\ref{thm:RWA}. By Assumption~\ref{assum:cvx}, \eqref{eq:AV1} holds by construction, and \eqref{eq:AV2} follows immediately since \(x\notin\mathcal A\) implies \(g(x)\ge 0\). To show \eqref{eq:AV3}, let \(x\notin\mathcal A\) be such that the MPC problem~\eqref{eq:RMPC:orig} is feasible, and let \(\{z_0^\star=x,z_1^\star,\dots,z_N^\star\}\) and \(\{u_0^\star,\dots,u_{N-1}^\star\}\) be a feasible solution. By the constraint \(z_1^\star\in\mathcal X^\RWA_{1}\) and the definition of \(\mathcal X^\RWA_{1}\), we have
\[
\min_{\bar w\in\bar{ W}} g(z_1^\star + C_0\bar w)\ge 0.
\]
Since \(C_0\bar w = Cw_0\) and \(w_0\) ranges over \( W\), this gives
\[
\min_{w\in W} g(z_1^\star + Cw)\ge 0.
\]
Using \(z_1^\star = A x + B u_0^\star = Ax + B\kappa^\RWA_{\mathrm{MPC}}(x)\), it follows that
\[
g(Ax+B\kappa^\RWA_{\mathrm{MPC}}(x)+Cw)\ge 0,
\qquad
\forall w\in W.
\]
Thus, condition~\eqref{eq:AV3} holds with the feedback law \(u=\kappa^\RWA_{\mathrm{MPC}}(x)\) whenever the MPC controller is applied. Therefore, \(g\) satisfies the Avoid conditions of Theorem~\ref{thm:RWA}.

\smallskip
4) Evaluating the cost of the shifted feasible sequence and using optimality yields
\begin{align*}
V_\RWA(z_1^\star)
&\le
T_\RWA(\tilde z_N)
+
\sum_{i=1}^{N-1} \ell_\RWA(z_i^\star,u_i^\star)
+
\ell_\RWA(z_N^\star,\kappa_f^\RWA(z_N^\star))
\\
&\le
V_\RWA(x) - \ell_\RWA(x,u_0^\star).
\end{align*}
The second inequality follows from Assumption~\ref{Assum:terminal}, since
\[
T_\RWA(\tilde z_N) - T_\RWA(z_N^\star)
\le
-\ell_\RWA(z_N^\star,\kappa_f^\RWA(z_N^\star)).
\]
Consequently, the MPC value function decreases along the nominal closed-loop evolution outside the target set \(\mathcal R\).

Since $x_{\mathrm{ref}}$ is chosen in the interior of $\mathcal R$, there exists $\psi>0$ such that $B_\psi(x_{\mathrm{ref}})\subseteq \mathcal R$, where $B_\psi(x_{\mathrm{ref}})$ is the closed Euclidean ball centered at $x_{\mathrm{ref}}$ with radius $\psi$. Therefore, for every $x\notin\mathcal R$ we have $\|x-x_{\mathrm{ref}}\|_2\ge \psi$, and hence the quadratic stage cost satisfies $\ell_\RWA(x,u)\ge \lambda_{\min}(Q_\ell)\psi^2=:c_\RWA>0$ for all admissible inputs $u$ outside the target set \(\mathcal R\). Therefore,
\begin{equation}
V_{\RWA}\bigl(Ax + B\kappa_{\mathrm{MPC}}^{\RWA}(x)\bigr) - V_{\RWA}(x) \le -c_{\RWA},
\label{eq:nominal_decrease}
\end{equation}
for all $x \notin \mathcal R$. For any fixed $x\in X$, define
\[
\phi_{\RWA}(w)
:=
V_{\RWA}\bigl(Ax + B\kappa_{\mathrm{MPC}}^{\RWA}(x) + Cw\bigr) - V_{\RWA}(x).
\]
By~\eqref{eq:nominal_decrease}, we have $\phi_{\RWA}(0)\le -c_{\RWA}$. Since $V_{\RWA}$ is continuous and $W$ is compact, the function $\phi_{\RWA}$ is continuous in $w$. Therefore, there exists $\eta_{\RWA}>0$ such that
\[
\phi_{\RWA}(w)\le -\frac{c_{\RWA}}{2},
\qquad
\forall\, w\in W \text{ with } \|w\|\le \eta_{\RWA}.
\]
Define
\[
\delta_{\RWA}:=\frac{c_{\RWA}}{2},
\qquad
\varepsilon_{\RWA}:=\mathbb P_w(\|w\|\le \eta_{\RWA})>0.
\]
Then condition~\eqref{eq:RE1} follows immediately by the disturbance assumption in Assumption~\ref{assum:cvx}. Therefore, \(V_\RWA\) can be used as the basis of a reachability certificate. More precisely, define
\[
r(x):=
\begin{cases}
V_\RWA(x), & x\notin\mathcal R,\\
-\rho(x), & x\in\mathcal R,
\end{cases}
\]
where \(\rho(x)\ge 0\) is any continuous function satisfying \(\rho(x)=0\) on \(\partial\mathcal R\) and \(\rho(x)>0\) in the interior of \(\mathcal R\). Since the stage and terminal costs are quadratic and the prediction horizon is finite, the value function \(V_\RWA\) is continuous on its feasible set. Furthermore, because the state space \(X\) is bounded, \(V_\RWA\) is locally bounded. By construction, \(r\) is therefore locally bounded and consequently, \(r\) satisfies the reach conditions of Theorem~\ref{thm:RWA}.
\end{proof}

To guarantee the Stay condition after reaching \(\mathcal R\), we adopt a switching strategy. While \(x_k\notin \mathcal R\), the robust MPC~\eqref{eq:RMPC:orig} is applied. Once the state enters \(\mathcal R\), we switch to a local controller designed to robustly enforce forward invariance of \(\mathcal R\).

For the Stay part, a local control is obtained by enforcing the robust invariance condition as
\begin{subequations}\label{eq:stay1}
\begin{align}
\min_{u\in U}\quad &
\|A x + B u - x_{\mathrm{ref}}\|_{Q_s}^2 \\
\text{s.t.}\quad &
\sup_{w\in W} h(A x + B u + C w) \le 0,
\end{align}
\end{subequations}
for some \(Q_s\succ 0\), with optimal solution denoted by \(\kappa_s^\RWA(x)\). Note that under Assumption~\ref{assum:target_invariance}, optimization~\eqref{eq:stay1} is feasible for all $x\in\mathcal{R}$.

\begin{corollary}\label{cor:RWA_switch}
Under Assumption~\ref{Assum:terminal}, the feedback law
\begin{equation}
    \kappa^\RWA(x)
=
\begin{cases}
\kappa^\RWA_{\mathrm{MPC}}(x), & x\notin \mathcal R,\\
\kappa^\RWA_s(x), & x\in \mathcal R
\end{cases}, \label{eq:control}
\end{equation}
satisfies the Avoid, Reach, and Stay conditions of
Theorem~\ref{thm:RWA}. Consequently, the cRAS objective
\(\RWA\) is implemented almost surely. By Lemma~\ref{lemma:reduction}, if every generated cRAS
objective is implemented by a controller of the form
\eqref{eq:control}, then the resulting hybrid state-feedback policy
\(\policy_{\mathrm h}\) solves Problem~\ref{prob:MainProb}.
\end{corollary}

\begin{proof}
For \(x\notin\mathcal R\), the controller coincides with \(\kappa^\RWA_{\mathrm{MPC}}\). By Theorem~\ref{thm:MPC}, this controller enforces the Avoid conditions of Theorem~\ref{thm:RWA}. Moreover, the MPC value function \(V_\RWA\) provides the Reach certificate, so the Reach conditions of Theorem~\ref{thm:RWA} are satisfied as well.

For \(x\in\mathcal R\), the controller switches to \(\kappa_s^\RWA\). By the construction of the local stay controller and Assumption~\ref{assum:target_invariance}, the target set \(\mathcal R\) is robustly forward invariant, so the Stay conditions of Theorem~\ref{thm:RWA} hold. Therefore, all conditions of Theorem~\ref{thm:RWA} are satisfied, and the cRAS objective \(\RWA\) is implemented almost surely.

By Lemma~\ref{lemma:reduction}, if every generated cRAS objective is implemented by a controller of the form~\eqref{eq:control}, then the resulting hybrid state-feedback policy \(\policy_{\mathrm h}\) solves Problem~\ref{prob:MainProb}.
\end{proof}
The certificate framework provided in Theorem~\ref{thm:RWA} plays a central role in the proposed architecture. The high-level synthesis procedure described in Section~\ref{sec:reduction} produces a collection of context-triggered cRAS objectives, whereas the low-level controller operates directly on the continuous-state dynamics~\eqref{eq:dyn}. The certificates introduced in Theorem~\ref{thm:RWA} provide the interface between these two layers: they translate the logical requirements of a cRAS objective into verifiable properties of the continuous closed-loop system, namely avoidance, reachability, and invariance. Theorem~\ref{thm:MPC} shows that the proposed robust MPC formulation naturally generates the Avoid and Reach certificates, while the local invariant controller enforces the Stay condition. Consequently, Corollary~\ref{cor:RWA_switch} establishes that the resulting switching controller implements the cRAS objective. Combined with Lemma~\ref{lemma:reduction}, this implies correctness of the overall hybrid controller. 

While the robust MPC formulation~\eqref{eq:RMPC:orig} explicitly enforces avoidance for all disturbance realizations, it involves a minimization over disturbance sequences inside the constraints, which may be computationally expensive in its current form. To obtain computationally tractable synthesis procedures, the next section employs convex duality to reformulate the robust constraints into equivalent deterministic optimization problems.

\section{Tractable Reformulation via Convex Duality}\label{sec:tract}
While the certificate-based MPC construction developed in Section~\ref{sec:C_RMPC} provides a correct realization of the cRAS objectives, the resulting optimization problems are not directly tractable. In particular, the robust avoidance and invariance constraints involve nested optimization problems over disturbance realizations. Moreover, these disturbance optimizations are coupled across the prediction horizon through the lifted disturbance matrices \(C_i\), causing the dimension of the inner optimization to grow with the horizon length \(N\). As a result, directly enforcing the robust constraints can become computationally expensive even for moderate horizons.

To obtain computationally tractable synthesis procedures, we exploit convex duality to eliminate the inner disturbance optimization problems and replace them with equivalent deterministic constraints. This yields finite-dimensional optimization problems whose structure depends on the geometry of the avoidance and disturbance sets.

\subsection{Robust MPC Reformulation}

Under Assumption~\ref{assum:cvx}, we show that the inner minimization in~\eqref{eq:RMPC:orig} admits an exact dual representation. This allows us to eliminate the disturbance variables entirely and replace the robust constraints by deterministic constraints involving dual variables.

\begin{lemma}\label{lem:dual_identity}
Let \(g:\mathbb R^{n_x} \to \mathbb R \cup \{+\infty\}\) be proper, convex, and closed, and let \(\mathbb W \subset \mathbb R^{q}\) be convex and compact. Then for any \(z\) and matrix \(M\in\mathbb{R}^{n_x\times q}\), the identity
\begin{equation}\label{eq:PD}
\min_{w \in \mathbb W} g(z + M w)
=
\sup_{y \in \mathbb R^{n_x}}
\left(
y^\top z
-
g^*(y)
-
\sigma_{\mathbb W}(-M^\top y)
\right)
\end{equation}
holds, where \(g^*(y)\) is the convex conjugate of \(g\) defined as
\[
g^*(y)
:=
\sup_{t \in \mathbb R^{n_x}}
\left(
y^\top t - g(t)
\right),
\]
and \(\sigma_{\mathbb W}(v):=\max_{w \in \mathbb W} v^\top w\) is the support function of \(\mathbb W\).
\end{lemma}

\begin{proof}
By the Fenchel--Moreau theorem, \(g=g^{**}\) for a proper, convex, and closed function \(g\) (see, e.g., Theorem~4.2.1 in~\cite{borwein2006convex}), i.e.,
\[
g(z + Mw)
=
\sup_{y \in \mathbb{R}^{n_x}}
\left\{
f_z(w,y):= y^\top (z + Mw) - g^*(y)
\right\}.
\]
Substituting this into the left-hand side of~\eqref{eq:PD} yields
\[
\min_{w \in \mathbb W} g(z + M w)
=
\min_{w \in \mathbb W}\sup_{y \in \mathbb{R}^{n_x}} f_z(w,y).
\]
Note that \(f_z(w,y)\) is convex (affine) in \(w\). Moreover, it is concave in \(y\), because it is the sum of the concave function \(-g^*\) and a linear term in \(y\).

Sion's minimax theorem~\cite{sion1958general} states that for a function \(f(w,y)\) that is convex--concave, \(\min_w \sup_y f_z(w,y)=\sup_y \min_w f_z(w,y)\) if one of the sets (in this case \(\mathbb W\)) is compact. Since \( \mathbb W\) is compact and \(g^*\) is convex, the conditions hold and swapping the infimum and supremum gives
\[
\sup_{y \in \mathbb{R}^{n_x}}
\left\{
y^\top z - g^*(y) + \min_{w \in \mathbb W}(y^\top M w)
\right\}.
\]
Finally, the last term can be written as
\[
\min_{w \in \mathbb W}(M^\top y)^\top w
=
-\max_{w \in \mathbb W}(-M^\top y)^\top w
=
-\sigma_{\mathbb W}(-M^\top y),
\]
which proves~\eqref{eq:PD}.
\end{proof}
Geometrically, the dual variable \(y\) can be viewed as a certificate direction. For a given \(y\), the quantity \(y^\top z-g^*(y)\) measures how far the predicted state is from violating the avoidance constraint, while \(\sigma_W(-M^\top y)\) captures the worst-case effect of disturbances in that direction. The dual problem therefore searches for a direction \(y\) proving that the avoidance constraint remains satisfied despite all admissible disturbances. This replaces the explicit optimization over disturbances by an equivalent certificate-based condition. 

Applying Lemma~\ref{lem:dual_identity} with \(z=z_{i+1}\), \(M=C_i\), \(\mathbb W=\bar W\), and $q=Nn_w$, gives
\[
\min_{\bar w\in \bar W} g(z_{i+1}+C_i \bar w)
=
\sup_{y_i}
\left(
y_i^\top z_{i+1}
-
g^*(y_i)
-
\sigma_{\bar W}(-C_i^\top y_i)
\right).
\]
Therefore, the robust constraint in~\eqref{eq:RMPC:orig} is equivalent to
\[
\exists y_i \in \mathbb R^{n_x}:\quad
y_i^\top z_{i+1}
-
g^*(y_i)
-
\sigma_{\bar W}(-C_i^\top y_i)
\ge 0.
\]
From the definition of \(\bar W\), its support function decomposes additively:
\[
\sigma_{\bar W}(v)
=
\sum_{j=0}^{N-1}
\sigma_W(v_j),
\]
where \(v_j\) are the \(n_w\)-dimensional blocks of \(v\). Moreover, because \(C_i\) has only the first \(i+1\) nonzero blocks,
\[
\sigma_{\bar W}(-C_i^\top y_i)
=
\sum_{j=0}^{i}
\sigma_W\!\left(
(-C_i^\top y_i)_j
\right).
\]
Hence, the robust MPC in~\eqref{eq:RMPC:orig} is equivalent to the deterministic problem
\begin{subequations}\label{eq:RMPC2}
\begin{align}
\min_{z,u,y}\quad
&
T_\RWA(z_{N})
+
\sum_{i=0}^{N-1}
\ell_\RWA(z_{i},u_{i})
\\
\text{s.t.}\quad
&
z_{i+1}=A z_i + B u_i, \qquad i=0,\ldots,N-1,\nonumber\\
&
u_i \in  U, \qquad i=0,\ldots,N-1,\nonumber\\
&
z_0 = x,\nonumber\\
&
z_N \in \mathcal X^\RWA_{\mathrm f},\nonumber\\
& 
y_i^\top\! z_{i+1}
\!\!-\!
g^*\!(y_i)
\!-\!
\!\sum_{j=0}^{i}
\!\!\sigma_W\!\left(
(-C_i^\top y_i)_j
\right)
\!\ge\! 0,\label{eq:robustCons2}
\\ &\qquad\qquad i=0,\ldots,N-1.\nonumber
\end{align}
\end{subequations}
Consequently, the robust MPC formulation~\eqref{eq:RMPC:orig} is equivalent to the deterministic optimization problem~\eqref{eq:RMPC2}, in which the disturbance minimization is replaced by dual variables \(y_i\). The term \(g^*\) encodes the avoidance geometry, while the support functions encode the disturbance-set geometry.

While the dual reformulation derived above is exact under the convexity assumptions, its computational structure depends on the specific geometry of the avoidance function \(g\) and the disturbance set \(W\). Different choices lead to different classes of convex programs. We therefore consider two practically relevant cases in detail:
(i) ellipsoidal avoidance sets with polyhedral disturbances, and
(ii) polyhedral avoidance sets with ellipsoidal disturbances.

\subsection*{Case 1: Ellipsoidal Avoid Set and Polyhedral Disturbance}
We first consider the case where the avoidance set is ellipsoidal and the disturbance set is polyhedral.

Consider the avoidance set \(\mathcal A\) with function 
\[
g(x)= (x-x_a)^\top Q (x-x_a) - 1, \quad Q\succ 0,
\]
for some $x_a\in \mathcal{A}$. The convex conjugate of \(g\) is given by (see e.g.,~\cite{boyd2004convex})
\[
g^*(y) = x_a^\top y + \frac{1}{4} y^\top Q^{-1} y + 1.
\]
Moreover, consider the disturbance set is polyhedral, given by:
\begin{equation}\label{eq:polyW}
W = \{ w \in \mathbb R^{n_w} : H_W w \le h_W \},
\end{equation}
for some \(H_W\in \mathbb{R}^{n_p\times n_w}\) and \(h_W\in \mathbb{R}^{n_p}\). Its support function is the linear program
\[
\sigma_W(v)
=
\max_{w \in W} v^\top w,
\]
and by linear programming duality its dual problem is
\begin{equation}\label{eq:dual:case1}
\sigma_W(v)
=
\inf_{\mu \ge 0}
\left\{
h_W^\top \mu
\;\middle|\;
H_W^\top \mu = v
\right\}.
\end{equation}
Substituting \(g^*\) into~\eqref{eq:robustCons2}, we have
\[
y_i^\top (z_{i+1}-x_a)
-
\frac{1}{4} y_i^\top Q^{-1} y_i
-
\sum_{j=0}^{i}
\sigma_W\!\left(
(-C_i^\top y_i)_j
\right)
\ge 1.
\]
Using~\eqref{eq:dual:case1}, this is equivalent to the existence of multipliers \(\mu_{ij} \ge 0\) such that
\begin{subequations}\label{eq:case1cons}
\begin{align}
&H_W^\top \mu_{ij} = -(C_i^\top y_i)_j, \\
&y_i^\top (z_{i+1}-x_a)
-
\frac{1}{4} y_i^\top Q^{-1} y_i
-
\sum_{j=0}^{i}
h_W^\top \mu_{ij}
\ge 1.
\end{align}
\end{subequations}
Hence, replacing the robust constraint~\eqref{eq:robustCons2} by~\eqref{eq:case1cons} with the additional decision variables $\mu_{ij}\ge0$ in~\eqref{eq:RMPC2}, the robust constraint is equivalently expressed through quadratic and linear inequalities. 

\subsection*{Case 2: Polyhedral Avoid Set and Ellipsoidal Disturbance}
We now consider the complementary case in which the avoidance set is polyhedral and the disturbance set is ellipsoidal. Consider a polyhedral avoidance set with
\[
g(x) = \max_{k=1,\dots,p} (f_k^\top x - b_k).
\]
Its convex conjugate is given by (see e.g.,~\cite{boyd2004convex})
\[
g^*(y)
=
\begin{cases}
\lambda^\top b,
&
y = F^\top \lambda,
\quad
\lambda \ge 0,
\quad
\mathbf 1^\top \lambda = 1,
\\
+\infty,
&
\text{otherwise},
\end{cases}
\]
where $F :=
[
f_1^\top,\ldots,
f_p^\top
]^\top
\in \mathbb R^{p\times n_x}$ and $b := (b_1,\dots,b_p)^\top$.
Assume now that the disturbance set is ellipsoidal:
\begin{equation}\label{eq:ellip:dist}
W
=
\{ w \in \mathbb R^{n_w} : w^\top P^{-1} w \le 1 \},
\qquad P \succ 0.
\end{equation}
Then the support function of \(W\) is obtained as
\[
\sigma_W(v)
=
\max_{w}
v^\top w
\quad
\text{s.t.}
\quad
w^\top P^{-1} w \le 1.
\]
This is a convex quadratic program. Using the Cauchy--Schwarz inequality,
\[
v^\top w
=
(P^{1/2} v)^\top (P^{-1/2} w)
\le
\|P^{1/2} v\|_2
\|P^{-1/2} w\|_2.
\]
Since \(w^\top P^{-1} w \le 1\) implies \(\|P^{-1/2} w\|_2 \le 1\), the maximum is attained when \(w\) is aligned with \(P v\). Therefore,
\[
\sigma_W(v)
=
\|P^{1/2} v\|_2
=
\sqrt{v^\top P v}.
\]
Substituting into~\eqref{eq:robustCons2}, the robust constraint becomes
\[
y_i^\top z_{i+1}
-
g^*(y_i)
-
\sum_{j=0}^{i}
\sqrt{
(C_i^\top y_i)_j^\top
P
(C_i^\top y_i)_j
}
\ge 0.
\]
Since \(g^*(y_i)\) is finite only when
\[
y_i = F^\top \lambda_i,
\qquad
\lambda_i \ge 0,
\qquad
\mathbf 1^\top \lambda_i = 1,
\]
the constraint becomes
\[
\lambda_i^\top (F z_{i+1} - b)
-
\sum_{j=0}^{i}
\sqrt{
(C_i^\top F^\top \lambda_i)_j^\top
P
(C_i^\top F^\top \lambda_i)_j
}
\ge 0.
\]
Introducing auxiliary variables \(t_{ij}\) satisfying
\begin{equation}\label{eq:cons1:case2}
\| P^{1/2} (C_i^\top F^\top \lambda_i)_j \|_2 \le t_{ij},
\end{equation}
the constraint can be written as
\begin{equation}\label{eq:cons2:case2}
\lambda_i^\top (F z_{i+1} - b)
-
\sum_{j=0}^{i} t_{ij}
\ge 0,
\qquad
\lambda_i \ge 0,
\qquad
\mathbf 1^\top \lambda_i = 1,
\end{equation}
which is second-order cone representable.

The size of the resulting optimziation grows polynomially with the prediction horizon. In particular, for each stage \(i\), one introduces a dual variable \(\lambda_i\) and \(i+1\) auxiliary variables \(t_{ij}\), together with the corresponding second-order cone constraints~\eqref{eq:cons1:case2}. Consequently, the total number of auxiliary variables and conic constraints scales as \(\mathcal O(N^2)\) with the horizon length \(N\).

Note that additional geometric configurations can be handled analogously through the same duality framework. For brevity, we restrict our attention to the two representative cases above, which already capture the most common safety and disturbance models encountered in practice.

\subsection{Switching Strategy Reformulation}
We now derive a tractable dual characterization of the Stay condition in~\eqref{eq:stay1} under a quadratic target function and a polyhedral disturbance set. Suppose the target function is quadratic and defined as
\[
h(x) := \|x - x_{\mathrm{ref}}\|_{Q_h}^2 - \rho_h^2,
\]
for some \(Q_h \succ 0\) and \(\rho_h > 0\). Assume the disturbance set is polyhedral as in~\eqref{eq:polyW}. Using convex duality for the maximization over disturbances, the robust constraint in~\eqref{eq:stay1} can be equivalently expressed through the existence of dual variables \(\mu \ge 0\) such that
\begin{equation}\label{eq:stay2}
h_W^\top \mu
-
\frac{1}{4} \mu^\top H_W H_W^\top \mu
+
\xi^\top H_W^\top \mu
\le \rho_h^2,
\end{equation}
where $\xi := A x + B u - x_{\mathrm{ref}}.$  Similar to the MPC robust constraints, the inner maximization over disturbances in~\eqref{eq:stay1} can therefore be reformulated using convex duality, yielding a tractable deterministic condition. The derivation follows the same convex-duality arguments used in Section~5.1 and is therefore omitted for brevity. In particular, applying Lemma~\ref{lem:dual_identity} to the quadratic stay certificate \(h\) and the disturbance set \(W\) yields an equivalent deterministic reformulation of the robust invariance constraint. Therefore, both the robust MPC problem~\eqref{eq:RMPC:orig} and the local invariance problem~\eqref{eq:stay1} admit equivalent deterministic reformulations obtained through convex duality.

\section{Numerical Simulations}\label{sec:Sim}

In this section, we demonstrate the proposed framework on a robot navigation problem subject to temporal logic constraints. We consider both static and dynamic environments and compare the resulting feasibility regions with the Lyapunov-based framework of~\cite{nayak2023context}. We consider the discrete-time stochastic linear system in~\eqref{eq:dyn} with $A = I_2$, $B = I_2$ and $C = I_2$,
and let the disturbance set be the bounded box
\begin{equation*}
 W = \{ w \in \mathbb{R}^2 \mid \|w\|_\infty \leq 0.03 \},
\end{equation*}
and the admissible control inputs satisfy
\[
U=\{u\in\mathbb R^2\mid\|u\|_\infty\leq 0.15\}.
\]
The workspace is a rectangular domain $[0,3] \times [0,3]$ partitioned into two rooms connected by a doorway. The MPC prediction horizon is $N=6$ and its stage and terminal costs follow~\eqref{eq:stage_terminal} with $Q_\ell=I_2$, $R_\ell=0.1I_2$, and $Q_T=10I_2$. The local controller follows~\eqref{eq:stay1} with $Q_s=I_2$. All simulations were implemented in MATLAB using CasADi and solved online with IPOPT~\cite{andersson2018casadi}. The environment contains three circular target regions $\mathcal T_1$, $\mathcal T_2$, and $\mathcal T_3$, each of radius $r=0.3$, defined as
\begin{equation*}
\mathcal T_i = \{ x \in \mathbb{R}^2 \mid \|x - t_i\|_2 \leq r \}, \quad i \in \{1,2,3\}.
\end{equation*}

The target centers are:
\begin{equation*}
t_1 = \begin{bmatrix}0.5 \\ 1 \end{bmatrix}, \quad
t_2 = \begin{bmatrix}1 \\ 1.75 \end{bmatrix}, \quad
t_3 = \begin{bmatrix}2.25 \\ 1.5 \end{bmatrix}.
\end{equation*}

A wall with a width of $0.3$ separates the two rooms, with a door in the middle of width $1$ that can either be open or closed. 
Starting in \(\mathcal T_1\), the robot must repeatedly navigate between the target regions while reacting to logical context switches and changes in the door status. In particular, reaching \(\mathcal T_2\) closes the door, while reaching \(\mathcal T_1\) or \(\mathcal T_3\) reopens it. The high-level temporal specification triggers context-dependent RAS tasks that, in this scenario, unfold as a switching sequence of $8$ phases:
\begin{enumerate}
    \item Starting at $\mathcal T_1$ with the door open, reach $\mathcal T_3$ while avoiding $\mathcal T_2$.

    \item After reaching $\mathcal T_3$ and dwelling for $3$ time steps, reach $\mathcal T_2$ while avoiding $\mathcal T_1$.

    \item Upon reaching $\mathcal T_2$, the door closes, and the task is to remain in $\mathcal T_2$.

    \item After $3$ time steps, reach $\mathcal T_1$ while avoiding $\mathcal T_3$ and the door (which is closed).

    \item Before reaching $\mathcal T_1$, a preemptive switch occurs: reach $\mathcal T_2$ while avoiding $\mathcal T_1$, $\mathcal T_3$, and the door.

    \item After reaching $\mathcal T_2$ and dwelling for $3$ time steps, reach $\mathcal T_1$ while avoiding $\mathcal T_3$.

    \item Upon reaching $\mathcal T_1$, the door opens, and the task is to reach $\mathcal T_3$ while avoiding $\mathcal T_2$.

    \item After reaching $\mathcal T_3$ and dwelling for $3$ time steps, reach $\mathcal T_1$ while avoiding $\mathcal T_2$, and terminate upon reaching $\mathcal T_1$.
\end{enumerate}

We first consider the case where all targets remain fixed over time. Figure~\ref{fig:1} illustrates the robot trajectory under the proposed controller together with the environment geometry, where the purple regions depict robust reachable tubes generated by the MPC predictions, ensuring safety under all disturbance realizations. The corresponding control inputs are shown in Figure~\ref{fig:2}, remaining within the prescribed bounds.

\begin{figure*}[t]
    \centering
    \includegraphics[width=\textwidth]{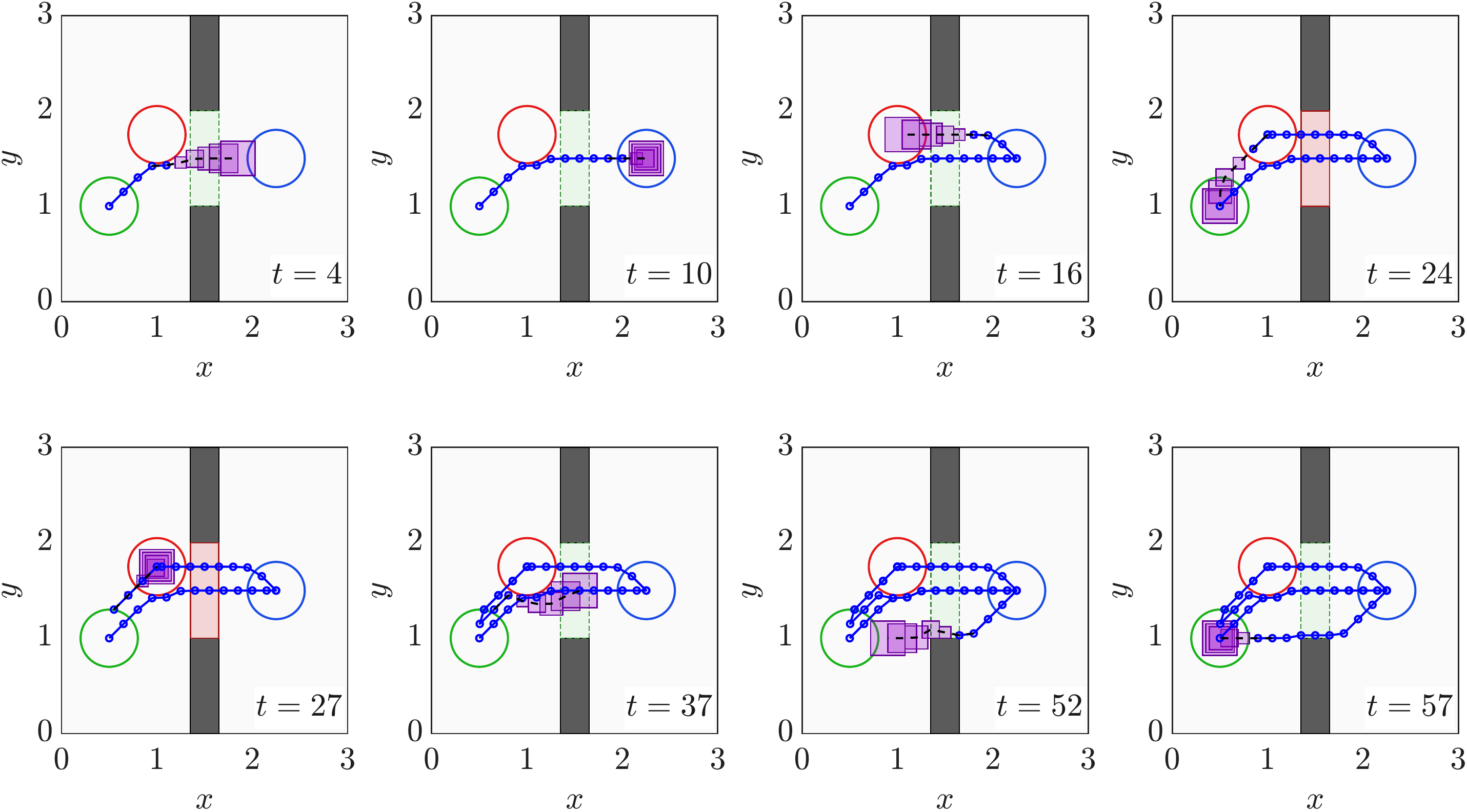}
    \caption{Closed-loop trajectory and robust tubes in the static environment. The robot successfully satisfies all RAS tasks under bounded disturbances.}
    \label{fig:1}
\end{figure*}

\begin{figure}[]
    \centering
    \includegraphics[width=0.48\textwidth]{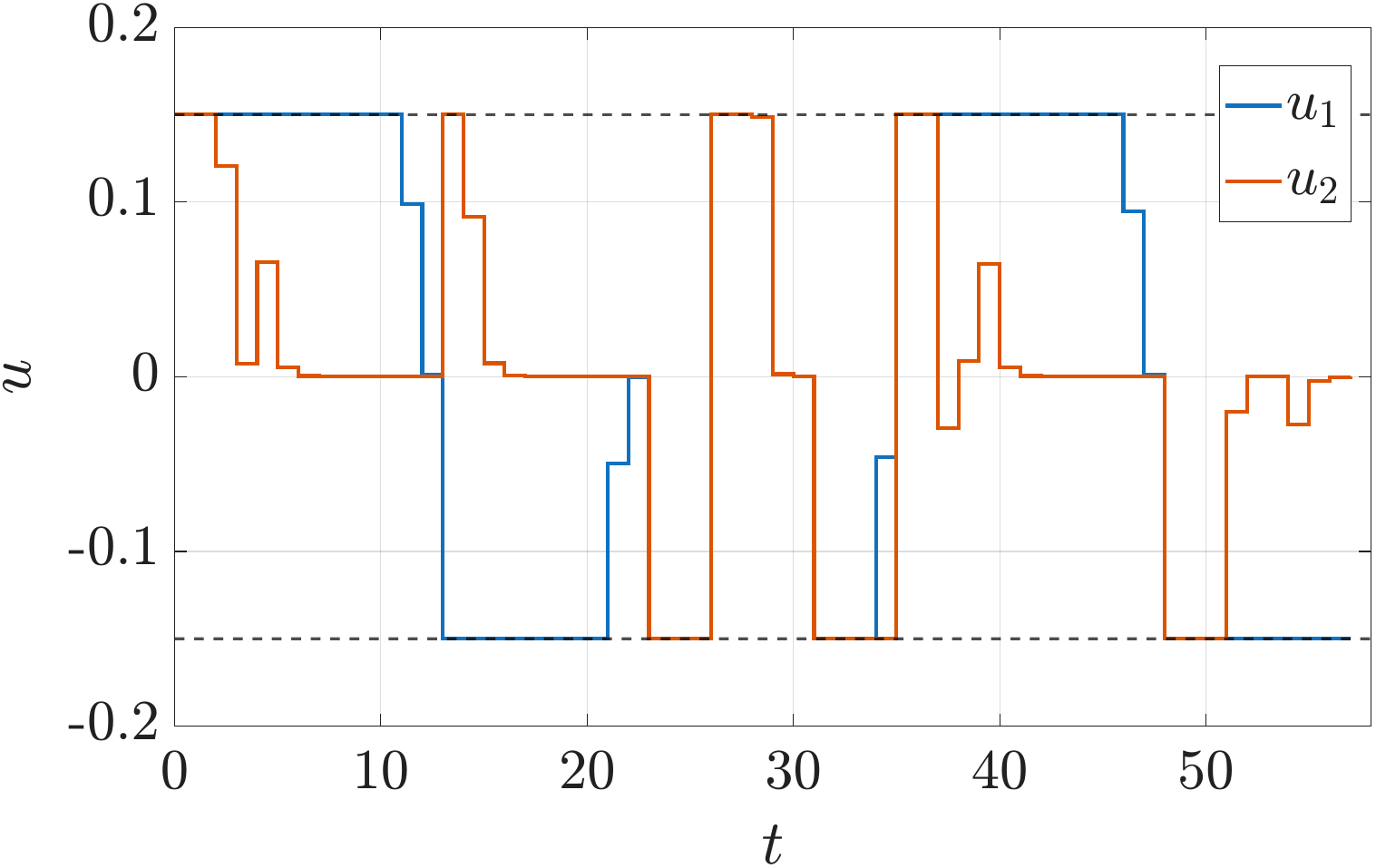}
    \caption{Control inputs for the static scenario. The inputs respect the constraints.}
    \label{fig:2}
\end{figure}

We next consider a more challenging scenario in which the targets \(\mathcal T_2\) and \(\mathcal T_3\) move over time, resulting in a dynamic environment. Figure~\ref{fig:3} shows the resulting trajectories. Despite the moving targets and stochastic disturbances, the robot continuously adapts its motion and satisfies all RAS specifications. The corresponding control inputs are shown in Figure~\ref{fig:4}, showing that both inputs respect the control constraints. 

\begin{figure*}[t]
    \centering
    \includegraphics[width=\textwidth]{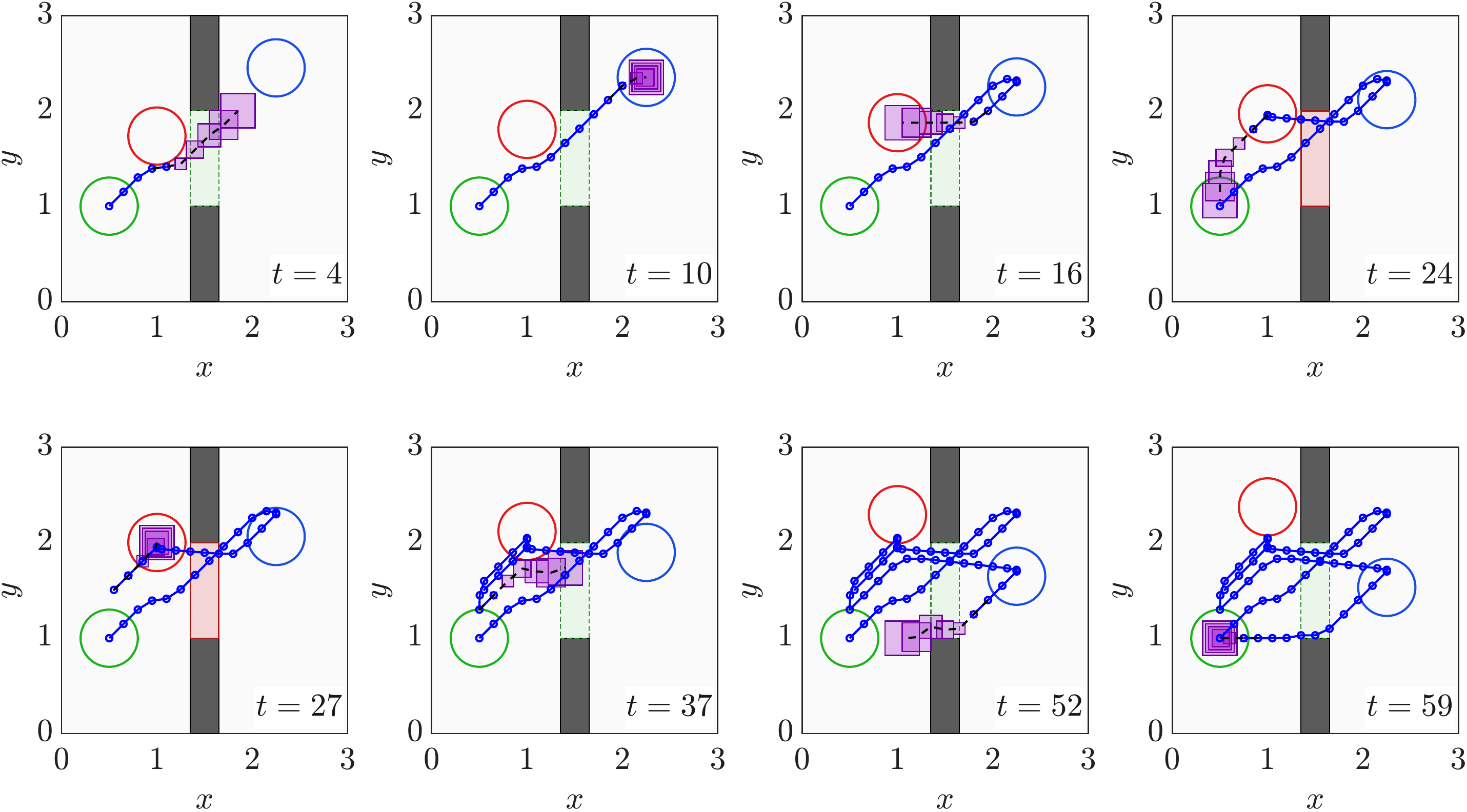}
    \caption{Closed-loop trajectory in the moving target scenario. The controller successfully tracks moving targets while maintaining safety and respecting the logical constraints.}
    \label{fig:3}
\end{figure*}

\begin{figure}[]
    \centering
    \includegraphics[width=0.48\textwidth]{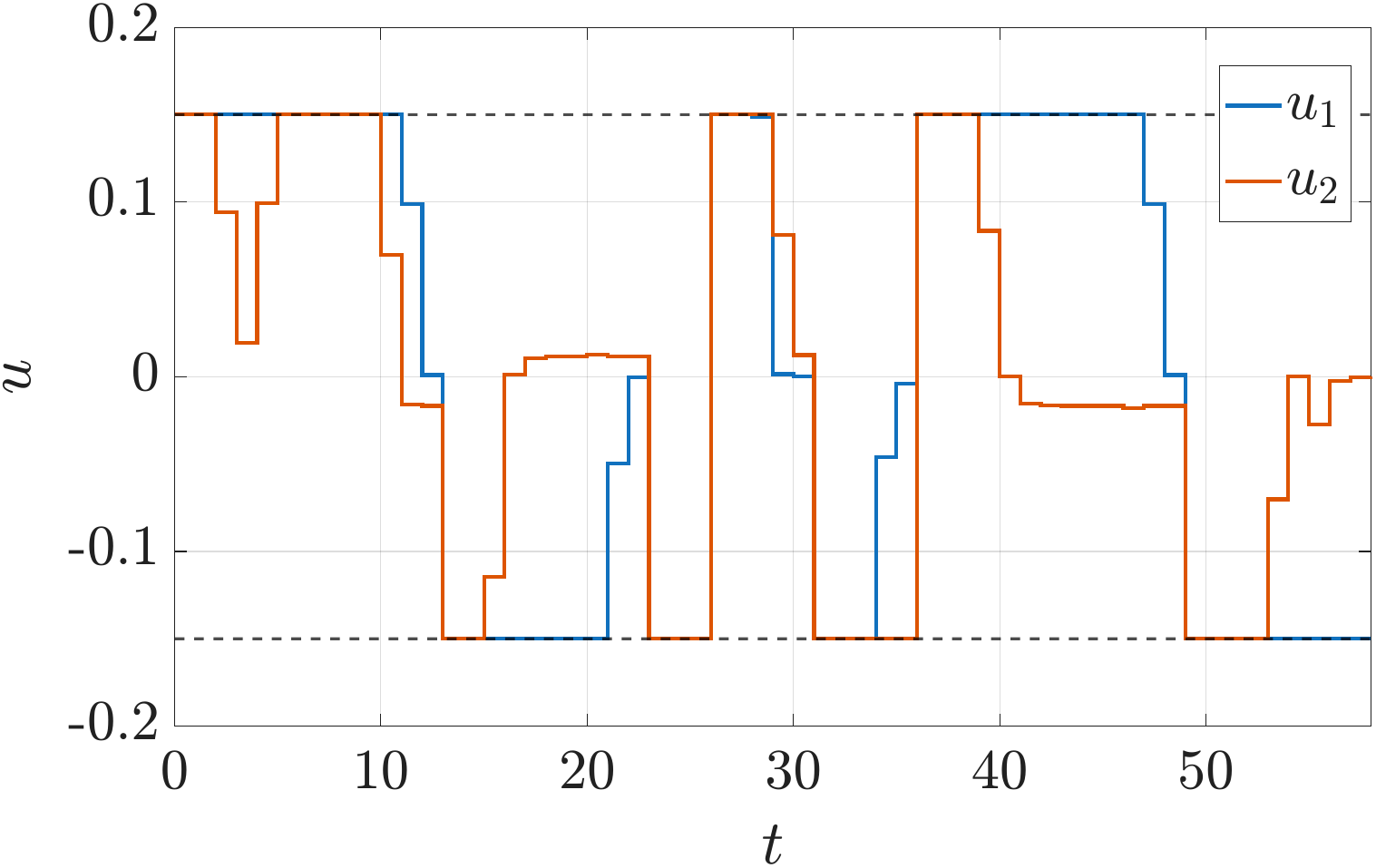}
    \caption{Control inputs in the moving target scenario. The MPC adapts to time-varying targets without violating constraints.}
    \label{fig:4}
\end{figure}

Figure~\ref{fig:5} compares the feasibility region of the proposed method with the Lyapunov-based approach in~\cite{nayak2023context}, which relies on control Lyapunov functions. To quantify the performance of the proposed approach, we compare the size of the feasible set of both approaches. This figure shows that the proposed robust MPC approach substantially enlarges the feasible set. In particular, in the closed-door scenario, the feasible region increases by a factor of $4.38$ (area ratio), while in the open-door scenario, the improvement is by a factor of $3.57$. Since the system considered in \cite{nayak2023context} is deterministic and fully actuated, the feasible set of the MPC formulation essentially coincides with the entire state space, excluding only the obstacle sets. In contrast, the approach of \cite{nayak2023context} relies on Lyapunov-function sublevel sets to characterize the region of attraction. Such sublevel sets, particularly when based on quadratic Lyapunov functions, are ellipsoidal and therefore provide only an inner approximation of the admissible state space, which can be significantly conservative. As a result, the MPC framework is able to exploit a substantially larger portion of the admissible state space, leading to the larger feasible sets shown in Figure~\ref{fig:5}.

\begin{figure}[]
    \centering
    \includegraphics[width=0.48\textwidth]{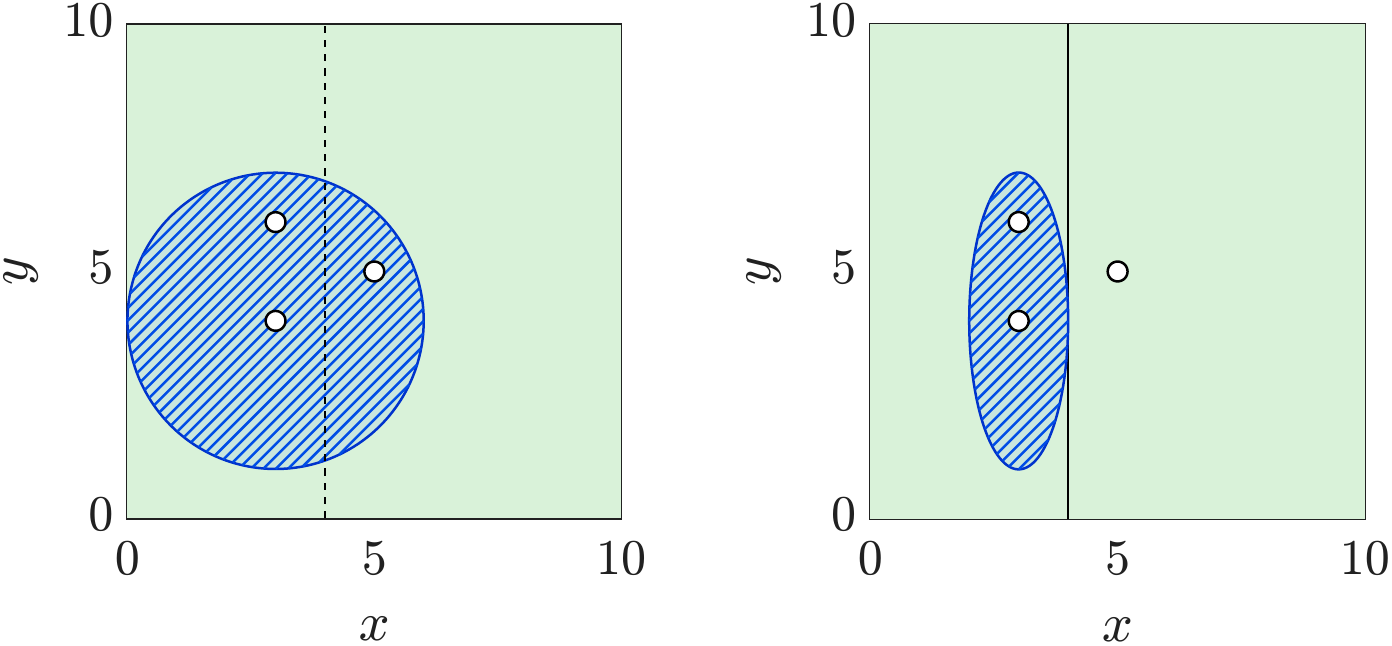}
    \caption{Comparison of the feasible sets obtained with the proposed robust MPC approach (green region) and the control Lyapunov function-based approach of \cite{nayak2023context} (blue hatched region) for the mobile robot navigation setting considered in \cite{nayak2023context}. Left: open-door scenario. Right: closed-door scenario.}
    \label{fig:5}
\end{figure}

Furthermore, while Lyapunov-based approaches are typically designed for static environments, the proposed framework naturally extends to time-varying settings, as demonstrated by the moving target scenario. This highlights not only the ability of the method to handle richer classes of specifications, but also its flexibility in adapting to dynamically changing environments.

The simulations also illustrate several important properties of the proposed approach. In particular, the controller maintains safety under all admissible disturbance realizations and achieves almost-sure reachability of the target sets. The switching behavior correctly reflects the logical phase transitions induced by the high-level specification, and the controller adapts seamlessly between static and dynamic scenarios. Together, these observations demonstrate that the proposed combination of certificate-based reasoning and robust MPC provides an effective and flexible framework for enforcing temporal logic specifications under uncertainty.

\section{Conclusions}\label{sec:Conc}
In this paper, we addressed the problem of synthesizing feedback controllers for discrete-time stochastic linear systems to ensure the robust satisfaction of context-triggered linear temporal logic specifications under additive disturbances. Focusing on context-dependent reach-avoid-stay (cRAS) objectives arising in context-triggered hybrid control frameworks, we introduced certificate-based sufficient conditions for robust safety, probabilistic reachability, and robust invariance. Based on these conditions, we proposed a control architecture combining robust finite-horizon optimal control with a local switching strategy, and showed that the resulting value function provides a reachability certificate. We further employed convex duality to reformulate the robust constraints into equivalent deterministic optimization problems, leading to tractable formulations for relevant geometric settings. The proposed method was demonstrated on a robot navigation example with logical context switches in both static and moving environments, where it achieved the desired RAS behavior under uncertainty, enlarged the feasible set compared with Lyapunov-based methods, and naturally extended to dynamic scenarios. These results provide a step toward integrating high-level temporal logic synthesis with low-level robust control for stochastic systems.

\bibliographystyle{model1-num-names}
\bibliography{CT_RMPC}

\end{document}